\documentclass[10pt,conference]{IEEEtran}
\usepackage{graphicx}
\usepackage{amsmath}
\usepackage{amsthm}
\usepackage{amsfonts}
\theoremstyle{}
\newtheorem{theorem}{Theorem}

\newtheorem{definition}{Definition}
\theoremstyle{remark}

\newtheorem{example}{Example}
\newtheorem{remark}{Remark}
\usepackage{bm}
\usepackage{booktabs}
\usepackage{multirow}
\usepackage{algorithm}
\usepackage{algpseudocode}
\usepackage{subfigure}
\makeatletter
\usepackage{cite}
\newcommand{\tabcaption}{\def\@captype{table}\caption}
\makeatletter
\ifCLASSINFOpdf
\else
\fi
\begin{document}
%
\title{Placement Delivery Array Design through Strong Edge Coloring of Bipartite Graphs }



%
\author{\IEEEauthorblockN{Qifa Yan$^{\dag}$, Xiaohu Tang$^{\dag}$, Qingchun Chen$^{\dag}$ and Minquan Cheng$^{\ddag}$}
\IEEEauthorblockA{$\dag$ School of Information Science and Technology, Southwest Jiaotong University, Chengdu, 611756, China\\
$\ddag$ Guangxi Key Lab of Multi-source Information Mining $\&$ Security, Guangxi Normal University, Guilin, 541004, China\\
Emails: qifa@my.swjtu.edu.cn, \{xhutang,~qcchen\}@swjtu.edu.cn, chengqinshi@hotmail.com
}
}



\maketitle

\begin{abstract} The technique of coded caching proposed by Madddah-Ali and Niesen is a promising approach to alleviate the load of  networks during busy times.  Recently, placement delivery array (PDA) was presented to characterize both the placement and delivery phase in a single array for the  centralized coded caching algorithm.  In this paper, we interpret PDA from a new perspective, i.e., the strong edge coloring of bipartite graph. We prove that, a PDA is equivalent to a strong edge colored bipartite graph. Thus, we can construct a class of PDAs from existing structures in bipartite graphs. The class includes the scheme proposed by Maddah-Ali \textit{et al.} and a more general class of PDAs proposed by Shangguan \textit{et al.} as special cases. Moreover, it is capable of generating a lot of PDAs with flexible tradeoff between the sub-packet level and load.
\end{abstract}


%
\IEEEpeerreviewmaketitle

\section{Introduction}
Video streaming has become the main impetus of the ongoing increasing wireless traffic.  In practice, excessive demands at busy times often lead to congestion in communication networks. One promising approach to mitigate the congestion is to duplicate some of contents that may be requested by users into  memories distributed across the network in advance, for example, at offpeak times \cite{Femtocaching2013}.

 In their seminal work on coded caching,  Maddah-Ali and Niesen proposed schemes that create multicasting opportunities by exploiting the caching \cite{maddah2013fundamental}. They modeled the content duplicating with a content placement phase, and the busy times with a content delivery phase. By jointly designing the two phases, the central server is able to deliver distinct contents to distinct users. Compared to conventional uncoded caching schemes, significant  gains can be achieved by coding in the delivery phase. The scheme is essentially centralized because the central server coordinates all the transmissions. The work has been extended to other scenarios, for example, decentralized  version \cite{maddah2013decentralized}, device to device networks \cite{Ji2016D2D}, online caching update \cite{maddah2013online} and hierarchical networks \cite{tree14}, \cite{Xiao2016tree} etc.

 In most of the existing coded caching schemes, each file has to be split into $F$ packets, where $F$ usually increases exponentially with the number of users $K$. This would be unpractical when $K$ is large. To study caching schemes that require smaller order of $F$, Yan \textit{et al.} reformulated the problem as a placement delivery array (PDA) design problem. A PDA clearly indicates what is cached by users and what is sent by the server in a single array. Thus, designing a centralized coded caching scheme for a given system is transformed into a problem of designing a proper PDA for given parameters. Accordingly, a construction was presented in \cite{Yan2015PDA}, which significantly decreases $F$, while the load only suffers from a slight sacrifice, in contrast to that in \cite{maddah2013fundamental}.

 Very recently, in view of PDA, Shangguan \textit{et al.} \cite{Ge2016hyper} extended the schemes in \cite{maddah2013fundamental} and \cite{Yan2015PDA} to a larger family of PDA. They showed that a PDA exists if and only if a corresponding linear $(6,3)$-free $3$-uniform $3$-partite hypergraph exists. They illustrated that, for some parameters, schemes that $F$ grows sub-exponentially with $K$ exist.

   In this paper, we establish the connection PDA with another topic in graph theory, i.e., strong edge coloring of bipartite graph.  We prove that a PDA exists if and only if a corresponding strong edge colored bipartite graph exists. This allows us to construct new PDAs from structured bipartite graphs. Particularly, we construct a class of PDA from bipartite graphs, which includes both the scheme in \cite{maddah2013fundamental} and its extended version in \cite{Ge2016hyper} as special cases. As a result, it contributes a larger cluster of PDA with more parameters. Particularly, it also generates some PDAs that are not covered by \cite{Ge2016hyper} such that $F$ increases sub-exponentially with $K$.

  The remainder  of this paper is as follows: In Section \ref{sec_prob}, we depict the system model and background on PDA.   In Section \ref{sec:graphPDA}, we introduce the strong edge coloring and its connection with PDA. Section \ref{sec:construction} depicts a new construction from bipartite graphs. Section \ref{sec:performance} presents some performance analysis of the new scheme. Finally, we conclude this paper in Section \ref{sec:con}.

\section{System Model and Background}\label{sec_prob}

\subsection{System Model}
Consider a caching system composed of a server containing $N$ files $\mathcal{W}=\{W_1,W_2,\cdots,W_N\}$ and    $K$ users $\mathcal{K}=\{1,2,\cdots,K\}$. The server is connected to the users through an error-free shared link, as illustrated in Fig. \ref{fig:system}.  The files in $\mathcal{W}$ are of equal size, without loss of generality (W.L.O.G), we assume each file is of unit size. Additionally, each user is equipped with a cache of size $M$.
\begin{figure}[htbp]
\centering\includegraphics[width=0.4\textwidth]{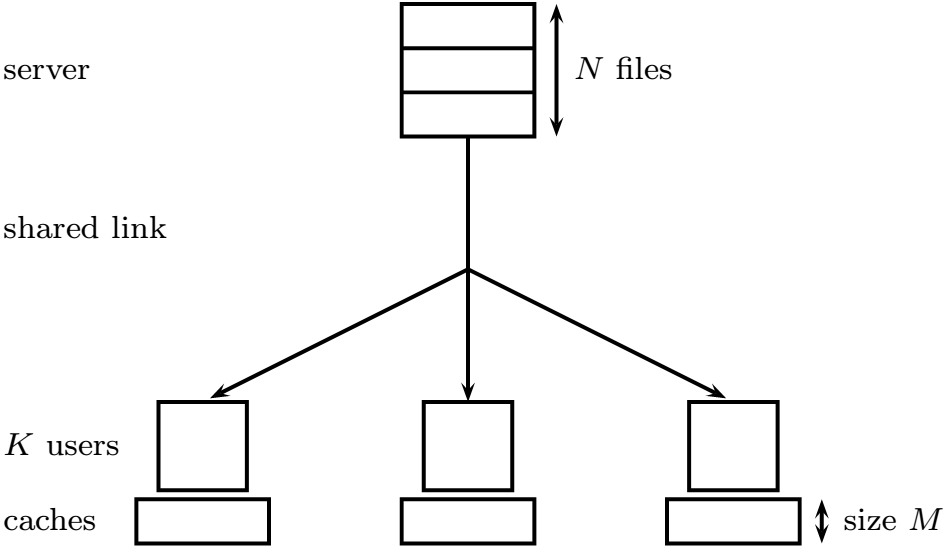}
\caption{A central server with $N$ popular files connected to $K$ users through an error-free shared link. Each file is of unit size  and each user is equipped with a cache of size $M$. In the figure, $N=K=3$, $M=1$.}\label{fig:system}
\end{figure}

We call such a system a $(K,M,N)$ caching system. This system operates in two separated phases:
\begin{enumerate}
  \item \textbf{Placement Phase:} A file is sub-divided into $F$ equal packets\footnote{Memory sharing technique may lead to non equally divided packets \cite{maddah2013fundamental}, in this paper, we will not discuss this case.}, i.e., $W_{i}=\{W_{i,j}:j=1,2,\cdots,F\}$. Each packet is of size $1/F$. Each user is accessible to the files set $\mathcal{W}$. Each packet is placed in different user caches determinately. Denote $\bm Z_k$ the cache contents of user $k$.
  \item \textbf{Delivery Phase:} Each user requests one file from $\mathcal{W}$ independently. Denote the request by $\bm d=(d_1,d_2,\cdots,d_K)$, i.e., user $k$ requests file $W_{d_k}$, where $k\in \mathcal{K}, d_k\in\{1,2,\cdots,N\}$. Once the server receives the request $\bm d$, it broadcasts a signal of at most $S_{\bm d}$ packets to users.  Each user is able to recover its requested file.
\end{enumerate}

 We refer to such a caching scheme as a $F$-division caching scheme, since each file is sub-divided into $F$ packets. Specially, the rate of the $F$-division caching scheme is defined as
 \begin{align}
 R\overset{\triangle}{=}\sup_{\bm d}\left\{\frac{S_{\bm d}}{F}\right\}\notag
 \end{align}
which should be as small as possible for a given  $(K,M,N)$ caching system.

\subsection{Background on Coded Caching and Placement Delivery Array}
The model was first formulated by Maddah-Ali and Niesen \cite{maddah2013fundamental}. For a $(K,M,N)$ caching system with $t\overset{\triangle}{=}{KM}/{N}\in\{0,1,\cdots,K\}$, they proposed a ${K\choose t}$-division caching scheme, which is called Ali-Niesen scheme in this paper. It turns out that the scheme can be conveniently depicted by  a single array defined as:
\begin{definition}(\textit{Placement Delivery Array}, \cite{Yan2015PDA})
For  positive integers $K,F, Z$ and $S$, a $F\times K$ array  $\bm{P}=[p_{j,k}]$, $1\leq j\leq F, 1\leq k\leq K$, composed of a specific symbol $``*"$  and $S$ positive integers
$1,2,\cdots, S$, is called a $(K,F,Z,S)$ placement delivery array (PDA) if it satisfies the following conditions:
\begin{enumerate}
  \item [C$1$.] The symbol $``*"$ appears $Z$ times in each column;
  \item [C2.] Each integer occurs at least once in the array;
  \item [C$3$.] For any two distinct entries $p_{j_1,k_1}$ and $p_{j_2,k_2}$,    $p_{j_1,k_1}=p_{j_2,k_2}=s$ is an integer  only if
  \begin{enumerate}
     \item [a.] $j_1\ne j_2$, $k_1\ne k_2$, i.e., they lie in distinct rows and distinct columns; and
     \item [b.] $p_{j_1,k_2}=p_{j_2,k_1}=*$, i.e., the corresponding $2\times 2$  subarray formed by rows $j_1,j_2$ and columns $k_1,k_2$ must be of the following form
  \begin{eqnarray*}
    \left[\begin{array}{cc}
      s & *\\
      * & s
    \end{array}\right]~\textrm{or}~
    \left[\begin{array}{cc}
      * & s\\
      s & *
    \end{array}\right]
  \end{eqnarray*}
   \end{enumerate}
\end{enumerate}
Furthermore, if each integer $s\in\{1,2,\cdots,S\}$ occurs exactly $g$ times, $\bm P$ is called $g$-$(K,F,Z,S)$ PDA, or $g$-PDA for short.
\end{definition}
With a given $(K,F,Z,S)$ PDA $\bm P=[p_{j,k}]$, a $F$-division caching scheme can be conducted by Algorithm \ref{alg:PDA}.
\begin{algorithm}[htb]
\caption{PDA based caching scheme}\label{alg:PDA}
\begin{algorithmic}[1]
\Procedure {Placement}{$\bm P$, $\mathcal{W}$}
\State Split each file $W_i\in\mathcal{W}$ into $F$ packets, i.e., $W_{i}=\{W_{i,j}:j=1,2,\cdots,F\}$.
\For{$k\in \mathcal{K}$}
\State \bm $\bm Z_k\leftarrow\{W_{i,j}:p_{j,k}=*, \forall~i=1,2,\cdots,N\}$
\EndFor
\EndProcedure
\Procedure{Delivery}{$\bm P, \mathcal{W},\bm d$}
\For{$s=1,2,\cdots,S$}
\State  Server sends $\bigoplus_{p_{j,k}=s,1\leq j\leq F,1\leq k\leq K}W_{d_{k},j}$.
\EndFor
\EndProcedure
\end{algorithmic}
\end{algorithm}

With Algorithm \ref{alg:PDA}, it was shown in \cite{Yan2015PDA} that the following fundamental theorem for PDA design.
\begin{theorem}(PDA based caching scheme, \cite{Yan2015PDA}) An $F$-division caching scheme for a $(K,M,N)$ caching system can be realized by a $(K,F,Z,S)$ PDA  with $Z/F=M/N$. Each user can decode his requested file correctly for any request $\bm d$ at the rate $R=S/F$.
\end{theorem}
 For a $(K,M,N)$ caching system,  with $M/N\in\{0,1/K,\cdots,1\}$, let $t=KM/N$, then the Ali-Niesen scheme is equivalent to a $(t+1)$-$\left(K,{K\choose t}, {K-1\choose t-1}, {K\choose t+1}\right)$ PDA \cite{Yan2015PDA}.

The scheme was extended to a more generalized construction by Shangguan \textit{et al.} in \cite{Ge2016hyper}:
 %
 For every three positive integers $m,a,b$, such that $a+b\leq m$, there exists an ${a+b\choose a}$-$\left({m\choose a},{m\choose b},{m\choose b}-{m-a\choose a},{m\choose a+b}\right)$ PDA.

\section{Strong Edge Coloring of Bipartite Graphs and Its Relation to PDA}\label{sec:graphPDA}
For clarity, we introduce several definitions from graph theory.

\begin{definition} Let $G$ be a graph.  $\mathcal{E}$ is the edge set of graph $G$. A (proper) edge coloring of $G$ is an assignment of colors to the edges $\mathcal{E}$ such that, no adjacent edges have the same color. The smallest number of colors needed in a (proper) edge coloring of a graph $G$ is the chromatic index of $G$.
\end{definition}

A strong edge coloring is actually a special edge coloring \cite{Fau1990strong}, i.e.,
\begin{definition} For a graph $G$ with edge set $\mathcal{E}$,  a strong edge coloring of $G$ is an edge coloring such that,  any pair of the edges with the same color are not  adjacent to any third edge in $\mathcal{E}$.
 The strong chromatic index $sq(G)$, is defined as the smallest number of colors in all possible strong edge colorings on $G$.
\end{definition}

In this paper, we focus on bipartite graphs. We denote a bipartite graph by $G(\mathcal{K},\mathcal{F}, \mathcal{E})$ where $\mathcal{K},\mathcal{F}$ are two disjoint vertical  sets and $\mathcal{E}\subset \mathcal{K}\times \mathcal{F}$ is the edge set.

Given a $F\times K$ array $\bm{A}=[a_{j,k}]$, composed of a specific symbol $``*"$ and integers $1,2,\cdots,S$ (W.L.O.G, we assume each integer occurs at least once in $\bm A$),   we construct a bipartite graph  $G(\mathcal{K},\mathcal{F},\mathcal{E})$  as follows:

Let $$\mathcal{K}=\{1,2,\cdots,K\}$$ $$\mathcal{F}=\{1,2,\cdots,F\}$$ and
\begin{align}
\mathcal{E}=\bigcup_{s=1}^S \mathcal{E}_s\label{eqn:E}
\end{align}
where
\begin{align}
\mathcal{E}_s=\{(k,j)\in\mathcal{K}\times\mathcal{F}:a_{j,k}=s \}\label{eqn:Es}
\end{align}
Then we color the edges in $\mathcal{E}_s$ by color $s$.

Conversely, given an edge colored bipartite graph $G(\mathcal{K},\mathcal{F},\mathcal{E})$, where $\mathcal{K}=\{1,2,\cdots,K\},\mathcal{F}=\{1,2,\cdots,F\}$, and  edges $\mathcal{E}$ are colored by colors $1,2,\cdots,S$, we can construct an $F\times K$ array $A=[a_{j,k}]$ composed of $``*"$ and $1,2,\cdots, S$ as follows:
\begin{align}
a_{j,k}=\left\{\begin{array}{cc}
                 * &  \mbox{if}~ (k,j)\notin \mathcal{E} \\
                 s &  \mbox{if}~ (k,j) \mbox{ is colored by}~ s.\
               \end{array}
\right.\label{eqn:ajk}
\end{align}

Obviously, an array $\bm A$ composed of $``*"$ and integers $1,2\cdots,S$ can be mapped to a colored bipartite graph in this way, and vice versa. An interesting result corresponding to PDA is:
\begin{theorem}\label{thm:bigraph}
For any $F\times K$ array $\bm A=[a_{j,k}]$, composed of a special symbol $``*"$ and integers $1,2,\cdots,S$, $\bm A$ is a PDA if and only if in its corresponding colored bipartite graph $G(\mathcal{K}, \mathcal{F},\mathcal{E})$,
\begin{enumerate}
  \item \label{item1} The verticals in $\mathcal{K}$ has a constant degree;
  \item \label{item2} The corresponding coloring is a strong edge coloring on $G(\mathcal{K},\mathcal{F},\mathcal{E})$.
\end{enumerate}
\end{theorem}
\begin{proof}
Assume $\bm A$ is a $(K,F,Z,S)$ PDA. Then there are $Z$ $``*"$  and $F-Z$ integers in each column by C$1$. Thus, by the edge construction  \eqref{eqn:E} and \eqref{eqn:Es}, for each given $k\in\mathcal{K}$, there are $F-Z$ elements in $\mathcal{F}$ such that $(k,j)\in\mathcal{E}$, i.e., the verticals in $\mathcal{K}$ have the same degree.

Moreover, we claim that for any two distinct edges with the same color $s$, i.e., $(k_1,j_1), (k_2,j_2)\in \mathcal{E}_s$,
they are not adjacent or adjacent to a third edge in $\mathcal{E}$.  By \eqref{eqn:Es}, we have
\begin{align}
p_{j_1,k_1}=p_{j_2,k_2}=s\notag
\end{align}
Thus, by C$3$-a, $k_1\neq k_2$, $j_1\neq j_2$. Thus, $(k_1,j_1), (k_2,j_2)$ are not adjacent to each other. Furthermore, by C$3$-b,
\begin{align}
p_{j_1,k_2}=p_{j_2,k_1}=*\notag
\end{align}
 by the construction of $\mathcal{E}$, i.e., \eqref{eqn:E} and \eqref{eqn:Es},  the edges $(k_1,j_2),$ $(k_2,j_1)$ are not in $\mathcal{E}$. Obviously, the edges $(k_1,k_2)$ and $(j_1,j_2)$ are not in $\mathcal{E}$ since $G$ is a partite graph. Thus, $(k_1,j_1)$ and $(k_2,j_2)$ are not adjacent to any third edge in $\mathcal{E}$. Thus, the corresponding coloring is a strong edge coloring on $G(\mathcal{K},\mathcal{F},\mathcal{E})$.

Conversely, given a colored bipartite graph $G(\mathcal{K},\mathcal{F},\mathcal{E})$, satisfying \ref{item1}),~\ref{item2}), we only need to verify that C$1$-C$3$ holds for $\bm A$. Since each $k\in\mathcal{K}$ has constant degree, for example $m$, there are $F-m$ verticals in $\mathcal{F}$ not connected to $k$. Thus, by \eqref{eqn:ajk}, there are $F-m$ $``*"$s in each column of $\bm A$, i.e., C$1$ holds with $Z=F-m$. C$2$ is trivial with the assumption each color at least colors one edge. For C$3$, assume two distinct elements $a_{j_1,k_1}=a_{j_2,k_2}=s$, then by \eqref{eqn:ajk}, two distinct edges $(k_1,j_1)$, $(k_2,j_2)$ are colored by $s$. Since the coloring is a strong edge color, $(j_1,k_1)$ and $(j_2,k_2)$ are not adjacent to each other, and there are no a third edge connects them, i.e., $k_1\neq k_2, j_1\neq j_2$ and $(k_1,j_2), (k_2,j_1)\notin\mathcal{E}$, thus by \eqref{eqn:ajk},
\begin{align}
a_{j_2,k_1}=a_{j_1,k_2}=*\notag
\end{align}
\end{proof}

Theorem \ref{thm:bigraph} establishes a connection between PDAs and  strong colored bipartite graphs. The advantage is that, we can utilize any existing result on the strong chromatic index of bipartite graphs in the literature to find new constructions of PDA, and furthermore to recover more coded caching schemes. We will discuss more about this in Section \ref{sec:construction}.

 We use the following example to illustrate Ali-Niesen scheme,  its relation to PDA and the corresponding strong colored bipartite graph.
\begin{example}
For a $(4,2,4)$ caching system, $t=2$. The PDA corresponding to Ali-Niesen scheme   is given by
\begin{align}
\bm A^{4,2}=\left[\begin{array}{cccc}
  * & * & 1 & 2 \\
  * & 1 & * & 3 \\
  * & 2 & 3 & * \\
  1 & * & * & 4 \\
  2 & * & 4 & * \\
  3 & 4 & * & *
\end{array}\right]\label{eqn:A}
\end{align}
Recall that the rows in $\bm A^{4,2}$ correspond to sets $\{1,2\},$ $\{1,3\},$ $\{1,4\},$ $\{2,3\},$$\{2,4\},$$\{3,4\}$ (see \cite{maddah2013fundamental}, \cite{Yan2015PDA} for details). Thus, $\bm A^{4,2}$  can be represented by a strong edge colored bipartite graph in Fig. \ref{fig:A42}.
\begin{figure}[htbp]
\centering\includegraphics[width=0.45\textwidth]{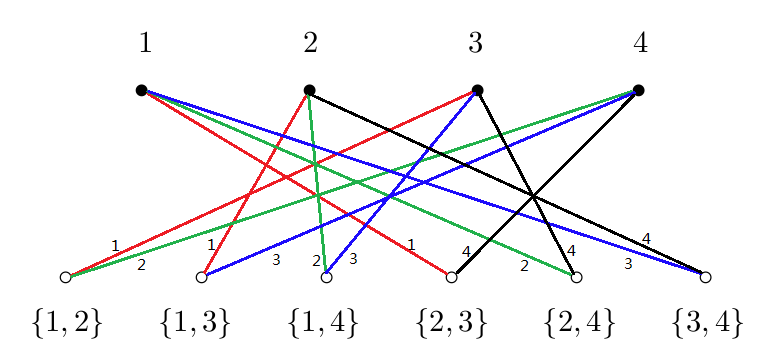}
\caption{The strong edge colored bipartite graph corresponding to PDA $\bm A^{4,2}$. The solid verticals correspond to users and the empty verticals correspond to packets. The edges are colored with $4$ colors, i.e., 1-red; 2-green; 3-blue; 4-black.  } \label{fig:A42}
\end{figure}

Each file $W_{i}$ ($i=1,2,3,4$) is divided into $6$ packets, i.e.,  $W_{i,\{1,2\}}, W_{i,\{1,3\}}, W_{i,\{1,4\}},W_{i,\{2,3\}},W_{i,\{2,4\}},W_{i,\{3,4\}}$ corresponding to the $6$ rows in $\bm A^{4,2}$ or the empty circles in Fig. \ref{fig:A42}. Thus, in the  placement phase, the contents in the users are
\begin{align}
\bm Z_1=\{W_{i,\{1,2\}},W_{i,\{1,3\}},W_{i,\{1,4\}},~i=1,2,3,4 \}\notag\\
\bm Z_2=\{W_{i,\{1,2\}},W_{i,\{2,3\}},W_{i,\{2,4\}},~i=1,2,3,4 \}\notag\\
\bm Z_3=\{W_{i,\{1,3\}},W_{i,\{2,3\}},W_{i,\{3,4\}},~i=1,2,3,4 \}\notag\\
\bm Z_4=\{W_{i,\{1,4\}},W_{i,\{2,4\}},W_{i,\{3,4\}},~i=1,2,3,4 \}\notag
\end{align}
i.e.,  a user vertical $k$ (a solid circle) store the packet $j$ if and only if $j$ is not connected to $k$ in Fig. \ref{fig:A42}.

In the delivery phase, W.L.O.G, we assume $\bm d=(1,2,3,4)$, then Table \ref{table1} lists the transmitting steps:
   \begin{table}[!htp]
  \centering
  \caption{Deliver steps of $\bm A^{4,2}$} \label{table1}
  \normalsize{
  \begin{tabular}{cc}
    \toprule  \toprule
   Time Slot& Transmitted Signnal  \\
   \midrule
   $1$&$W_{1,\{2,3\}}\oplus W_{2,\{1,3\}}\oplus W_{3,\{1,2\}}$\\
   $2$&$W_{1,\{2,4\}}\oplus W_{2,\{1,4\}}\oplus W_{4,\{1,2\}}$\\
  $3$& $W_{1,\{3,4\}}\oplus W_{3,\{1,4\}}\oplus W_{4,\{1,3\}}$\\
   $4$& $W_{2,\{3,4\}}\oplus W_{3,\{2,4\}}\oplus W_{4,\{2,3\}}$\\
 \bottomrule \bottomrule
  \end{tabular}}
\end{table}
\end{example}

\section{Constructions of PDA from Strong Edge Colored Bipartite Graphs}\label{sec:construction}
The edge coloring is an active research area. Unfortunately, for the strong chromatic index of bipartite graphs, there are rare results. For general case, only some small parameters results exist, for example,\cite{note2008} and \cite{Jul2016} discuss the strong edge coloring of bipartite graph whose one part of verticals has only at most $2$ or $3$ degrees. However, some specific structures do exist. In this section, we construct a class of PDAs  from the existing structures in bipartite graph.

Define $S_m(a,b,\lambda)=G(\mathcal{X},\mathcal{Y},\mathcal{E})$ ($0\leq \lambda\leq \min\{a,b\}$) where the verticals of $\mathcal{X}$ are the $a$-subsets of $\{1,2,\cdots,m\}$ (i.e., the subsets of cardinality $a$), the verticals of $\mathcal{Y}$ are the $b$-subsets of $\{1,2,\cdots,m\}$. Two verticals $X\in\mathcal{X}$ and $Y\in\mathcal{Y}$, $X$ is adjacent to $Y$ if and only if $|X\cap Y|=\lambda$.

For $S_{m}(a,b,\lambda)$,  the authors in \cite{Jen1997subset} gave two coloring schemes:
\begin{enumerate}
  \item[S$1$.] Given disjoint sets $D,I\subset\{1,2,\cdots,m\}$, where $|D|=a+b-\lambda$, $|I|=\lambda$, color the edges
  \begin{align}
  \mathcal{E}_{D,I}=\{(X,Y)\in\mathcal{E}:X\backslash Y\cup Y\backslash X=D, X\cap Y=I\}\notag
  \end{align}
  with the same color.
  \item[S$2$.] Given disjoint sets $U,V\subset\{1,2,\cdots,m\}$, where $|U|=a-\lambda, |V|=b-\lambda$, color the edges
  \begin{align}
  \mathcal{E}^{U,V}=\{(X,Y)\in\mathcal{E}:X\backslash Y=U,Y\backslash X=V\}\notag
  \end{align}
  with the same color.
\end{enumerate}

 \cite{Jen1997subset} has proved that both S$1$ and S$2$ are strong edge colorings on $S_m(a,b,\lambda)$, which use
  \begin{align}
  {m\choose a+b-2\lambda}{m-(a+b-2\lambda)\choose \lambda}\label{eqn:S1}
  \end{align}
   and
 \begin{align}
 {m\choose a+b-2\lambda}{a+b-2\lambda\choose a-\lambda }\label{eqn:S2}
 \end{align}
  colors respectively.

Now, according to Theorem \ref{thm:bigraph}, it is easy to obtain the following conclusion:
\begin{theorem}\label{thm:subset} For any $m,a,b,\lambda\in \mathbb{N}^+$ with $0<a<m,0<b<m, 0\leq \lambda\leq\min\{a,b\}$, there exists a $g$-$(K,F,Z,S)$ PDA, where
\begin{align}
&K={m\choose a},~ F={m\choose b},~Z={m\choose b}-{a\choose \lambda}{m-a\choose b-\lambda}\notag\\
&S=\notag\\
&{m\choose a+b-2\lambda}\cdot\min\left\{{m-(a+b-2\lambda)\choose \lambda},{a+b-2\lambda\choose a-\lambda }\right\}\notag
\end{align}
and
\begin{align}
g=\max\left\{{a+b-2\lambda\choose a-\lambda},{m-(a+b-2\lambda)\choose \lambda}\right\}\notag
\end{align}
\end{theorem}
\begin{proof}
The arrays are constructed from the $S_m(a,b,\lambda)$ with coloring strategies S$1$ and S$2$. Obviously, in $S_m(a,b,\lambda)$, each vertical in $\mathcal{X}$ has a degree ${a\choose \lambda}{m-a\choose b-\lambda}$, thus with the fact that S$1$ and S$2$ are both strong edge colorings on $S_{m}(a,b,\lambda)$, by Theorem \ref{thm:bigraph}, the constructed arrays are actually PDAs.  As for parameters, $K=|\mathcal{X}|={m\choose a}$, $F=|\mathcal{Y}|={m\choose b}$. Since each vertical in $\mathcal{X}$ has a degree ${a\choose \lambda}{m-a\choose b-\lambda}$, $Z={m\choose b}-{a\choose \lambda}{m-a\choose b-\lambda}$.

For S$1$ and S$2$, the number they use are \eqref{eqn:S1} and \eqref{eqn:S2} respectively. We choose the PDA with smaller $S$. Finally, it is easy to see
\begin{align}
|\mathcal{E}_{D,I}|&={a+b-2\lambda\choose a-\lambda}\notag\\
 |\mathcal{E}^{U,V}|&={m-(a+b-2\lambda)\choose \lambda}\notag
\end{align}
thus choose the PDA with smaller $S$ is equivalent to choose the PDA with larger coding gain. This results the coding gain $g$.
\end{proof}
\begin{remark}  By \eqref{eqn:S1} and \eqref{eqn:S2}, when $\lambda=0$, \eqref{eqn:S1} is smaller than \eqref{eqn:S2}, while when $\lambda=\min\{a,b\}$, \eqref{eqn:S2} is smaller than \eqref{eqn:S1}. For a general $\lambda$, it needs to compare \eqref{eqn:S1} and \eqref{eqn:S2} to decide which strategy will leads to smaller $S$ (thus smaller rate).
\end{remark}
\begin{remark} Note that, when $m=K, a=1,b=t$, Theorem \ref{thm:subset} results a $(t+1)$-$\left(K,{K\choose t}, {K-1\choose t-1},{K\choose t+1}\right)$ PDA. This is exactly the Ali-Niesen scheme. Furthermore, when $\lambda=0$, Theorem \ref{thm:subset} results a ${a+b\choose a}$-$\left({m\choose a},{m\choose b},{m\choose b}-{m-a\choose b}, {m\choose a+b}\right)$ PDA. This is the exactly the Scheme 1 in \cite{Ge2016hyper}. Thus, Theorem \ref{thm:subset} includes both Ali-Niesen scheme and the extended version in \cite{Ge2016hyper} as special cases.
\end{remark}

\section{Comparison with Ali-Niesen Scheme}\label{sec:performance}

 Since we can hardly express $R$ and $F$ as a function of $K$, it seems hard to measure the performance of the scheme in Theorem \ref{thm:subset}. But   Theorem \ref{thm:subset} may contributes a lot of schemes with flexible tradeoff between $F$ and $R$. An example is as follows:

 Let $a=2,\lambda=1,m=2b$, then Theorem \ref{thm:subset} results a scheme for a $(K,M,N)$ system where $K=b(2b-1)$, $\frac{M}{N}=\frac{b-1}{2b-1}\approx\frac{1}{2}$. Note that, this construction is not included in \cite{Ge2016hyper}.  For such system, we compare Ali-Niesen scheme with the New scheme in Theorem \ref{thm:subset}.
    \begin{table}[!htp]
  \centering
  \caption{Comparison with Ali-Niesen scheme} \label{table2}
  \normalsize{
  \begin{tabular}{ccccc}
    \toprule\toprule
  Scheme & Ali-Niesen &New  \\
   \midrule
  $R$ &$\frac{b^2}{b^2-b+1}$&$b$ \\
  $F$&${b(2b-1)\choose b(b-1)}$ &  ${2b\choose b}$\\
 \bottomrule\bottomrule
  \end{tabular}}
\end{table}

It can be seen that, with an expense of a multiple of $b-1+\frac{1}{b}$, (less than $b$) in rate, $F$ can be decreased dramatically, from ${b(2b-1)\choose b(b-1)}$ to ${2b\choose b}$. Note that, by Stirling's formular, we can derive
\begin{align}
{b(2b-1)\choose b(b-1)}&\sim \frac{1}{\sqrt{2\pi}b}\sqrt{\frac{2b-1}{b-1}} \left(\frac{2b-1}{b}\right)^{b^2} \left(\frac{2b-1}{b-1}\right)^{b^2-b}\notag\\
&=\frac{1}{\sqrt{2\pi}b}\left(4+\frac{1}{b^2-b}\right)^{b^2}\cdot\left(\frac{b-1}{2b-1}\right)^{b-\frac{1}{2}}\notag\\
{2b\choose b}&\sim \frac{1}{\sqrt{\pi b}}\cdot4^{b}\notag
\end{align}

In summary, the new scheme  decreases $F$ from approximately $O(4^{b^2})$ to $O(4^b)$. Note that since $b=\frac{1+\sqrt{1+8K}}{2}\approx\sqrt{2K}$, this indicates that in the new scheme, $F$ increases sub-exponentially with $K$, but $R$ increases approximately the same as the square root of $K$.  Table \ref{table6} lists some numerical results.

  \begin{table}[!htp]\centering\caption{ Comparisons ($a=2,\lambda=1, m=2b, b=3,4,5$)}\label{table6}
        \normalsize{\begin{tabular}{cccc}   \toprule\toprule
         \multicolumn{1}{c}{$\left(K,\frac{M}{N}\right)$}&Performance&Ali-Niesen& New\\\midrule
        \multirow{2}{*}{$(15,\frac{2}{5})
        $} &$R$    & 1.2857&  3 \\
        &$F$ &   5005  &   20 \\\midrule
        \multirow{2}{*}{$(28,\frac{3}{7})
        $} &$R$    & 1.2307&  4\\
        &$F$ &   30421755  &   70 \\\midrule
        \multirow{2}{*}{$(45,\frac{4}{9})
        $} &$R$    & 1.1904&  5 \\
        &$F$ &    3169870830126 & 252   \\
   \bottomrule\bottomrule
        \end{tabular}}
\end{table}

\section{Conclusion}\label{sec:con}
In this paper, we bridge the placement delivery array in coded caching with the strong edge coloring of bipartite graph. It turns out that a PDA is equivalent to a strong edge colored bipartite graph. This connection allows us to utilize the existing structures in the bipartite graph area to construct new PDAs. Then a class of PDA is  constructed according to the (generalized) subset graph \cite{Jen1997subset}. This results in a more general class of PDA compared to the extended version of Ali-Niesen scheme in \cite{Ge2016hyper}.

 We believe that, this connection between PDA and strong edge coloring of bipartite graphs provides a new perspective to coded caching. It is possible to utilize bipartite graph to do further research on coded caching. Meanwhile, further progress on strong edge colorings of bipartite graph will help on discovering new coded caching schemes.


%
%
%



%
\bibliographystyle{IEEEtran}

\newpage



\end{document}